\documentclass[11pt]{article}

\newenvironment{hangref}
  {\begin{list}{}{\setlength{\itemsep}{4pt}
  \setlength{\parsep}{0pt}\setlength{\leftmargin}{+\parindent}
  \setlength{\itemindent}{-\parindent}}}{\end{list}}

\def\qed{\relax\ifmmode\hskip2em \fbox{ }\else\unskip\nobreak\hskip1em 
$\fbox{}$\fi}

\newsavebox{\theorembox}
\newsavebox{\lemmabox}
\newsavebox{\corollarybox}
\newsavebox{\propositionbox}
\newsavebox{\examplebox}
\newsavebox{\propertybox}
\savebox{\theorembox}{\bf Theorem}
\savebox{\lemmabox}{\bf Lemma}
\savebox{\corollarybox}{\bf Corollary}
\savebox{\propositionbox}{\bf Proposition}
\savebox{\examplebox}{\bf Example}
\savebox{\propertybox}{\bf Property}
\newtheorem{theorem}{\usebox{\theorembox}}
\newtheorem{lemma}[theorem]{\usebox{\lemmabox}}

\newtheorem{definition}{{\sc Definition}\rm }[section]
\newtheorem{definitions}[definition]{{\sc Definitions\rm }}

\newenvironment{proof}{{\noindent\bf Proof~}}{\(\qed\)\vspace*{\proofskip} }
\newlength{\proofskip}
\begin{document}

\begin{center}

{\LARGE 
A Fast Heuristic for Exact String Matching
}

% Authors and addresses:

\footnotesize

\mbox{\large Srikrishnan Divakaran}\\
DA-IICT, 
Gandhinagar, Gujarat, India 382007,
\mbox{
Srikrishnan\_divakaran@daiict.ac.in }\\[6pt]
\normalsize
\end{center}

% begin "double spacing" the text:
\baselineskip 20pt plus .3pt minus .1pt

% Here is the abstract:

\noindent 
\begin{abstract}
Given a pattern string $P$ of length $n$ consisting of $\delta$ distinct characters     and a query string $T$ of length $m$, where the characters of $P$ and $T$  are drawn from an alphabet $\Sigma$ of size $\Delta$,    the {\em exact string matching} problem consists of finding all occurrences of $P$ in $T$. For this problem, we present a     randomized heuristic    that in $O(n\delta)$  time   preprocesses $P$ to  identify $sparse(P)$, a  rarely occurring substring of $P$, and  then   use  it  to find all occurrences of $P$ in $T$ efficiently. This    heuristic  has  an expected   search time of $O(
\frac{m}{min(|sparse(P)|, \Delta)})$, where $|sparse(P)|$ is at least $\delta$.  We also show that for a pattern string $P$   whose characters are chosen  uniformly at random from an alphabet of size $\Delta$, $E[|sparse(P)|]$ is $\Omega(\Delta log (\frac{2\Delta}{2\Delta-\delta}))$.
\end{abstract}
\bigskip
% Here are the Keywords:
\noindent\hrulefill
% The body of the paper starts here:
%******************************************************
\section{Introduction}
Given a pattern string $P$ of length $n$ consisting of $\delta$ distinct    characters  and a query string $T$ of length $m$,
where the characters of $P$ and $T$  are drawn from an alphabet $\Sigma$ of size $\Delta$,    the {\em exact string matching} problem consists of finding all occurrences of   $P$ in $T$. This is a  fundamental  problem  with wide range of applications 
in Computer Science (used    in   parsers,  word    processors, operating systems,  web  search engines, image processing and natural language processing),    Bioinformatics and   Computational Biology (Sequence   Alignment and Database Searches). The algorithms  for  exact string matching   can  be  broadly categorized into the following categories: (1) character      based 
comparison algorithms, (2) automata  based     algorithms, (3) algorithms     based on bit-parallelism and (4) constant-space 
algorithms. In   this  paper, our      focus is on   designing    efficient         character based comparison algorithms for 
exact string matching. For a comprehensive survey of all categories of exact string matching algorithms, we refer the readers 
to Baeza-Yates[17], Gusfield[25], Charras and Lecroq [28], Crochemore et al [29] and Faro and Lecroq [30]. \newline \newline
A Typical character based comparison algorithm  can    be     described within the following  general   framework as follows:
\begin{itemize} 
\item [(1)]
First, initialize the search window to be the first $n$ characters of the query string $T$ (i.e. align the $n$ characters  of 
the pattern string $P$ with the  first $n$ characters of $T$). 
\item [(2)] 
Repeat the following until the search      window is  no  longer contained within the query string $T$:
\begin{quote} 
inspect  the aligned pairs in some order until there is either a mis-match in an  aligned pair or there is   a complete match 
among all the    $n$ aligned pairs. Then       shift the search window to the right. The order in which the aligned pairs are 
inspected and the length by  which the search window is shifted differs from one algorithm to another.
\end{quote} 
\end{itemize} 
The mechanism    that     the above framework provides is usually referred to as the {\em sliding window mechanism} [30, 31]. 
The algorithms that  employ the sliding window mechanism can be further  classified  based on the order in which they inspect 
the  aligned pairs into  the following broad categories: (1) left to right scan; (2) right to left scan; (3) scan in specific 
order, and (4) scan in   random  order or scan order is not relevant. The algorithms that inspect the aligned pairs from left 
to right are the most natural  algorithms; the algorithms that inspect the aligned pairs from right to left generally perform 
well in practice; the algorithms  that inspect the aligned pairs  in a  specific  order  yield  the  best theoretical bounds. 
For    a     comprehensive     description   of the exact string matching algorithms and access to an excellent framework for 
development, testing and analysis  of exact string matching algorithms, we refer the readers to the {\em SMART tool}  (string 
matching research tool) of Faro and Lecroq [31]. \newline \newline 
For algorithms that inspect aligned pairs from left to right, Morris and Pratt [1] proposed the first known       linear time 
algorithm. This algorithm was  improved by Knuth, Morris and Pratt [4] and requires $O(n)$    preprocessing time and a  worst 
case search time of at most $2m-1$ comparisons. For small pattern strings and reasonable probabilistic assumptions  about the 
distribution of characters in the query string, hashing [2,9] provides an $O(n)$ preprocessing  time and $O(m)$   worst  case 
search time solution. For  pattern  strings   that  fit  within a  word   of    main memory, Shift-Or [16,21] requires $O(n + 
\Delta)$ preprocessing time and a  search time of $O(m)$ and  can     also be easily adapted  to solve     approximate string  
matching problems. 
For algorithms that inspects aligned pairs from right to left,     Boyer-Moore algorithm [3] is one of the classic algorithms 
that  requires $O(n\delta)$   preprocessing time and a worst case search time of $O(nm)$ but in practice is  very fast. There 
are    several variants that simplify the Boyer-Moore algorithm and mostly  avoid its quadratic behaviour. Among the variants 
of  Boyer-Moore, the    algorithms of   Apostolico and Giancarlo [7,24],  Crochemore et al [13, 23] (Turbo BM), and   Colussi (Reverse Colussi) [12, 22]  
have   $O(m)$ worst case search time and are efficient in minimizing the number of character comparisons, whereas the   Quick 
Search [10], Reverse Factor [19], Turbo Reverse Factor [24], Zhu and Takaoka [8], and Berry and Ravindran [27] algorithms are very efficient in practice.
For    Algorithms that inspects the aligned pairs in a specific order, Two Way algorithm [13], Colussi [12], Optimal Mismatch 
and Maximal Shift [10], Galil and Giancarlo [18],    Skip Search , KMP Skip Search and Alpha Skip Search [26] are some of the well known algorithms.
Two    way algorithm was the first known linear time optimal space algorithm. The     Colussi      algorithm      improves the 
Knuth-Morris-Pratt algorithm and requires at most $\frac{3n}{2}$ text character comparisons in the worst case. The   Galil and Giancarlo 
algorithm    improves the Colussi algorithm in one special case which enables it to perform at most $\frac{4n}{3}$ text character
comparisons in the worst case.  
For Algorithms that inspects the aligned pairs in any order,the Horspool [5], Quick Search [10], Tuned Boyer-Moore [14], Smith [15] and Raita [20]
algorithms are some of   the   well known algorithms. All these algorithms have   worst case search time that is quadratic but 
are known to perform well in practice. 
\newline \newline 
{\bf Our Results}: In this paper, we present  a          simple randomized heuristic $R$  that essentially preprocesses $P$ in $O(n\delta )$   time to  identify $sparse(P)$, a rarely occurring substring of $P$ characterized by two characters of $P$, and   use it to find  occurrences of $P$ in $T$ efficiently. Heuristic $R$ has a worst case search time  of $O(mn)$     and an expected search time  of $O(\frac{m}{min(|sparse(P)|, \Delta)})$.  The   worst           case search time of $O(mn)$ is realized due to pathological  pattern instances (i.e. periodic patterns consisting of few distinct characters) that are   rare      and highly 
unlikely to occur in practice. However, our algorithm has a superior provable sub-linear  expected search time when   compared to existing algorithms. In addition, we          prove     that for a pattern string $P$ whose characters are chosen uniformly
at random from an alphabet of  size $\Delta$, $E[|sparse(P)|]$ is 
$\Omega(\Delta log (\frac{2\Delta}{2\Delta-\delta}))$,
 where $\delta$ 
is the number of distinct characters in $P$.  
 \newline \newline 
The  rest    of this paper   is    organized as follows: In Section $2$, we present our randomized heuristic  $R$.  In Section 
$3$, we present the analysis of our heuristic. In Section $4$, we show    for a pattern string $P$ whose characters are chosen uniformly at random from an alphabet of size $\Delta$,   we lower bound $E[|sparse(P)|]$.   Finally, in Section $5$ we present   our conclusions and future work.
\section{A Randomized Heuristic for Exact String Matching}
In   this section, we present a randomized Heuristic $R$ that   given a pattern string $P$ and  a query string $T$, finds  all 
occurrences of $P$ within $T$. First, we introduce some definitions that are essential for defining our Heuristic. 
\begin{definitions}
Given a   pattern      string $P$, and an ordered   pair of characters $u, v \in \Sigma$ (not necessarily distinct), we define $sparse^{(u,v)}(P)$, the    $2$-sparse     pattern of $P$ with respect to $u$ and $v$, to be the rightmost     occurrence of a substring of     $P$ of longest length that starts with $u$, ends with $v$, but does not contain $u$ or $v$      within it. If 
such a substring does not exist then $sparse^{(u,v)}(P)$ is defined to be the empty string.  We define $sparse(P)$ to   be the 
right most occurring among the longest $2$-sparse patterns of $P$.
\end{definitions}
\begin{definitions}
Given $sparse(P)$, the   longest   $2$-sparse pattern of $P$, we define $startc(P)$ and $endc(P)$ to be  the respective first and last characters of $sparse(P)$, and $startpos(P)$   and    $endpos(P)$  be the respective   indices of the first and last characters of $sparse(P)$ in $P$. For $c \in \Sigma$, if $c \in sparse(P)$,      $shift^{c}(P)$ is  the distance between  the rightmost     occurrence     of $c$ in $sparse(P)$ and the  last character of $sparse(P)$.  If $c$ is not present in $P$ then $shift^{c}(P)$ is set to $n$, the length of $P$. If $c$ is present in $P$ but not in   $sparse(P)$ then $shift^{c}(P)$ is set 
to $|sparse(P)|+1$.
\end{definitions}
{\bf Example $1$}: Let   $P = "abcabdacabdbb"$. From definition, we   can see $sparse^{(a,a)}(P) = "abda"$, $sparse{(a,b)}(P)="ab"$ (substring starting at      location $9$),  $sparse^{(a,c)}(P) = "abc"$,   $sparse^{(a,d)}= "abd"$ (substring starting at location $9$), $sparse^{(b,a)}(P) = "bda"$,  $sparse^{(b,b)}(P) = "bdacab"$,   $sparse^{(b,c)}(P) = "bdac"$,      $sparse^{(b,d)}(P) = "bd"$ (substring  starting at location $10$), $sparse^{(c,a)}(P) = "ca"$ (substring starting at location $8$), $sparse^{(c,b)}(P) =
"cab"$ (substring starting at location $8$) and $sparse^{(c,c)}(P) = "cabdac"$. Therefore, $sparse(P)="bdacab"$, $startc(P) = "b"$, $endc(P) = "b"$, $startpos(P) = 5$, and $endpos(P) = 10$. $shift^{a}(P)=1, shift^{b}(P)=0, shift^{c}(P)=2$, and $shift^{d}(P)=4$. \newline \newline 
{\bf BASIC IDEA}:  First, we preprocess $P$ to identify $sparse(P)$, {\it a rarely occurring substring of $P$   characterized 
by two characters in $P$}, and compute statistics of the occurrence of characters of $P$ in $sparse(P)$ relative to $endc(P)$ (the last character of $sparse(P)$). Then, during         the search  phase, we first   set the search window to be the first $n$ characters of    $T$ (i.e. align the $n$ characters of $P$ to the  first $n$ characters of $T$). Then, until   the search 
window reaches the end of $T$, we    repeatedly check      whether there is a match  between the first and    last characters 
of $sparse(P)$  (i.e. $startc(P)$ and $endc(P)$)  and their respective aligned characters $c$ and $d$   in the search window. The following three situations are possible:
\begin{itemize}
\item[(i)] [Type-1] {\em A mismatch between $ endc(P)$ and its corresponding aligned character $c$}: In this case the search 
window  is shifted  by  $shift^{c}(P)$ (i.e. the distance in  $sparse(P)$ between the  right most occurrence  of     $c$ and $endc(P)$) and then continues. 
\item[(ii)] [Type-2] {\em A match between $endc(P)$ and its corresponding aligned character $c$ but a mismatch   between the 
$startc(P)$ and its corresponding aligned character $d$}: In this case the search window is shifted by at least $|sparse(P)|$ 
and then continues. 
\item[(iii)] [Type-3] {\em a match between $endc(P)$ and $startc(P)$ with their respective aligned characters $c$ and   $d$}: 
 In this case $R$ invokes routine $Random-Match$         to verify an exact match between the $n$ characters of $P$ and their
respective aligned characters in the search window. $Random-Match$ inspects the  $n$ aligned pairs in    random order   until
it encounters a    mismatch or finds an exact match (i.e. match in    all the $n$ aligned pairs). If  it finds an exact match
then it reports  the starting location of the exact match in $T$ and then  shifts the search window by at least $|sparse(P)|$
and  continues. Otherwise, it shifts the search window by at least $|sparse(P)|$ and continues.
\end{itemize}
{\bf Randomized Heuristic $R$}
\begin{tabbing}
Input(s): \= (1) Pattern string $P$ of length $n$; \\
          \> (2) Query string $T$ of length $m$;\\
Output(s): The starting positions of the occurrences of $P$ in $T$; \\ 
{\em Preprocessing}: \= (1) \= Scan $P$ from left to right and compute \\
                     \>     \>      \= [a] for each $u, v \in P$, $sparse^{(u,v)}(P)$; \\
                     \>     \>      \> [b] $sparse(P) = |sparse^{(a,b)}(P)|=max_{u,v \in P}^{} |sparse^{(u,v)}(P)|$. \\
                     \> (2) From $sparse(P)$, compute $startc(P), endc(P), startpos(P)$, and $endpos(P)$. \\
                     \> (3) For $c \in \Sigma$, compute $shift^{c}(P)$. \\                     
{\em Search}:\= \\
          \> [1] \= [a] Set $i=0$ and $j=n$ [Search Window set to $[1..n]$] \\
          \>     \> [b] Set $\hat{j}= endpos(P)$ and $\hat{i} = startpos(P)$; \= [Indices of last and first characters \\
          \>     \>                                                           \>  of $sparse(P)$ in $P$] \\
          \> [2] \= while $(j  <  m)$ [While the search window is contained in $T$] \\
          \>     \> [a] Let $c = T[i+\hat{j}]$; $d=T[i+\hat{i}]$; \= [Characters in search window aligned with \\
          \>     \>                                               \>  the last and first characters of $sparse(P)$]\\
          \>     \> [i] \= if  \= ($c \ne endc(P)$)  [Type-1 event] \\
          \>     \>     \>     \>  $i = i + shift^{c}(P)$; $j = j + shift^{c}(P)$; [Shift window by $shift^{c}(P)$] \\
          \>     \> [ii] \= else\=if ($d \ne startc(P)$) [Type-2 event] \\
          \>     \>     \>     \> if (\=$startc(P) == endc(P)$) \\
           \>    \>     \>     \>     \> $i=i+|sparse(P)|$; $j=j+|sparse(P)|$ [Shift window by $|sparse(P)|$] \\
           \>     \>     \>     \> else\= \\ 
           \>     \>     \>     \>     \> $i=i+|sparse(P)|+1$; $j=j+|sparse(P)|+1$ [Shift window by $|sparse(P)|+1$]\\
          \>     \> [iii]\= else\= [Type-3 event] \\
          \>     \>     \>     \> Call $Random-Match(P[1..n], T[i+1..i+n])$; [Look for $P$ in $T[i+1..j]$]\\
          \>     \>     \>     \> if (\=$startc(P) == endc(P)$) \\
          \>     \>     \>     \>     \> $i=i+|sparse(P)|$; $j=j+|sparse(P)|$ [Shift window by $|sparse(P)|$] \\
          \>     \>     \>     \> else\= \\ 
          \>     \>     \>     \>     \> $i=i+|sparse(P)|+1$; $j=j+|sparse(P)|+1$.[Shift window by $|sparse(P)|+1$] \\
\end{tabbing}
Subroutine {\bf Random-Match}
\begin{tabbing}
Input(s): \= Strings $P$ and $Q$ each of length $n$; \\
Output(s): If there is an exact match between the strings $P$ and $Q$ then return "TRUE" else "FALSE"; \\
begin \= \\
      \> Let $\pi = (\pi(1), \pi(2), ..., \pi(n))$ be a permutation drawn randomly from the set of $n!$ permutations \\
      \> of integers $1$ to $n$; \\
      \> for \= $(i=1; i \le n; i=i+1)$ \\
      \>     \> if \= $P[\pi(i)] \ne Q[\pi(i)]$ \\
      \>     \>    \> return FALSE; \\
      \> return TRUE; \\
end 
\end{tabbing} 
{\bf Remark}: We can view the Heuristic $R$ as using the first and last characters of $sparse(P)$ to decompose $T$    in  
$O(\frac{m}{min(|sparse(P)|, \Delta)})$
comparisons   into at most $\frac{m}{|sparse(P)|}$ substrings of $T$ of length $n$    that are 
potential candidates for an exact match with $P$. Then each of these $n$ candidate substrings are inspected for an exact 
match with $P$ using $Random-Match$ in $O(1)$ expected number of comparisons.
\section{Analysis of Algorithm $R$} 
In this section, we present the analysis of the randomized heuristic $R$. For rest of this paper we assume that the pattern string $P$ is an arbitrary string of length $n$ consisting of $\delta$ distinct characters that are drawn from an alphabet $\Sigma$ of size $\Delta$ and query string $T$ is a string of length $m$ whose characters are drawn from $\Sigma$ uniformly at random.  We now present the main results and their proofs.
\begin{theorem}
Randomized Heuristic $R$ finds all occurrences of $P$ in $T$.
\end{theorem}
\begin{proof}
The Algorithm $R$ during its search phase looks for $P$ by first looking for a match between   the last and first   characters of      $sparse(P)$ and the characters of $T$ in the search window  at offsets $endpos(P)$ and $startpos(P)$ respectively. Let 
$c$ and $d$           be       the respective characters in the search window at offsets  $endpos(P)$ and $startpos(P)$.   The 
following      three  scenarios (events) are possible. (i){\em Type-1} event  happens if there  is a mismatch between $c$  and
$endc(P)$  (ii) {\em Type-2} event happens if there is  a match between $c$  and $endc(P)$  and a mismatch  between  $d$   and  $startc(P)$, and (iii) {\em Type-3} event happens when there is a match between $c$    and   $endc(P)$ and a match between $d$ and $startc(P)$. \newline \newline 
In the case of Type-1 event, the  search  window is shifted to the   right by $shift^{c}(P)$. Recall that if $c$ is present in 
$sparse(P)$ then  $shift^{c}(P)$ is the distance in $sparse(P)$ between the right most occurrence of $c$  and $endc(P)$.    If 
$c$ is present in $P$ but  not in   $sparse(P)$ then $shift^{c}(P)$ is set to $|sparse(P)|+1$, otherwise $shift^{c}(P)$ is set 
to $n$, the length of $P$.     Now, by  shifting      the search window by $shift^{c}(P)$, we  will show that no occurrence of $sparse(P)$ will be skipped  and hence  no   occurrence of   $P$ will be skipped. There are two situations possible  depending 
on whether or not $endc(P)$ occurs in $T$  within the shifted interval. If $endc(P)$ did  not occur within     the     shifted interval of $T$ then we can see that no occurrence of $sparse(P)$ can end within the shifted interval. Hence  we will not skip $sparse(P)$ and hence not skip $P$. Now, we  consider    the situation when $endc(P)$ occurs within the shifted interval. This 
implies that $c$ occurs in the  search window less than $shift^{c}(P)$ positions to the right end of the search window.   From  definition of $shift^{c}(P)$, we know that the   right most     occurrence       of $c$ in $sparse(P)$ will be  $shift^{c}(P)$ positions to  the left of $endc(P)$.  This implies    that no occurrence of $sparse(P)$ can end  within the shifted portion of 
the search window. Hence we are done. \newline \newline 
In the case of $Type-2$ event, we know that $c == endc(P)$ but $d \ne startc(P)$ and the search window is shifted to the right by $|sparse(P)|$ ($|sparse(P)|+1$) if $startc(P) == endc(P)$ ($startc(P) \ne endc(P)$). Notice in this case, from   definition 
of $sparse(P)$ we know that the character $endc(P)$ does not occur inside $sparse(P)$, hence $sparse(P)$ cannot end within the 
shifted portion of     the search window. Hence we are done. \newline \newline 
In    the case of $Type-3$ event, we  can observe that $c ==endc(P)$ and $d = startc(P)$ and after invoking the $Random-Match$ subroutine we   shift    the window to the right      by either $|sparse(P)|$ ($|sparse(P)| + 1$) if    $startc(P) == endc(P)$
($startc(P) \ne endc(P)$). Notice  in this case, from    definition of $sparse(P)$ we know that the characters $endc(P)$   and 
$startc(P)$ do    not occur  inside $sparse(P)$, hence $sparse(P)$ cannot end within the shifted portion of the search window. Hence we are done.
\end{proof} 
\begin{theorem}
Given   any pattern string $P$ of length $n$ with $\delta$ distinct characters and a query string $T$ of length $m$, Heuristic $R$ finds all occurrences of $P$ in $T$ in $O(\frac{m}{min(|sparse(P)|, \Delta)})$
expected time, where  $|sparse(P)|$ is at least $\delta$.
\end{theorem}
\begin{proof} 
We bound the expected number of comparisons performed by Heuristic $R$ during its search phase  by looking at the expected number of comparisons performed during each of the three type of events it encounters in comparison to the expected length by which the   search window is shifted. For Type-1 events the number of character comparisons involving the query string $T$ is $1$ and for Type-2 events it is $2$. For Type-3 events, after two character comparisons we are invoking the Random-Match subroutine. From Lemma $4$, we know the expected number of matches before a mismatch while invoking $Random-Match$ is $O(1)$. Therefore the number of character comparisons for    Type-3 event is $O(1)$. Now since the number of character comparisons involving $T$ is $O(1)$ for all three events, the total number of comparisons during the search phase of Heuristic $R$ is bound by the total number of events.  From   Lemma $3$, we know that the expected  length by which the search window is shifted after encountering a Type-1,  Type-2 or Type-3 event is at least $min(|sparse(P)|,\Delta)$. Since the total length by which the search window can be shifted is $m$, the total number of events is therefore $O(\frac{m}{min(|sparse(P)|, \Delta)})$. From Lemma $5$ we know that $|sparse(P)|$ is at least $\delta$, and from definition we know $\Delta \geq \delta$, therefore $ min(|sparse(P)|,\Delta) \geq \delta$. Hence the result.
\end{proof} 
\begin{lemma}
For   any    pattern   string $P$, during the search phase of Heuristic $R$, the expected length of shift of the search window after a Type-1, Type-2 or Type-3 event is at least $min(|sparse(P)|,\Delta)$.
\end{lemma}
\begin{proof}
Notice that the query string $T$ is of length $m$ and each of its characters are drawn uniformly    at random from the alphabet $\Sigma$ of size $\Delta$. In the case of a Type-1 event, we can observe that the search window  is shifted by $shift^{c}(P)$.    Notice       that the mismatch character $c$ is equally likely to be any character in $\Sigma$   other  than $endc(P)$. So, we can observe that if $c \in sparse(P)$, then $shift^{c}(P)$ is equally likely to be any value in the  interval $[1..\delta]$ and if $c \notin sparse(P)$ then $shift^{c}(P)$ is at least $|sparse(P)|$. Therefore, the expected  shift  length will be at least $[(1 + 2 +...+ \delta) + (\Delta-\delta)|sparse(P)|)]/\Delta$. If $\delta < \Delta/2$, then this above  sum is $O(|sparse(P)|)$, otherwise it is $O(|\Delta|)$. Therefore, the expected shift length is at least $O(min(|sparse(P)|,\Delta))$. In the case of Type-2 event, we can observe that the search window is shifted by at least $|sparse(P)|$. Similarly    after   a type-3 event, Random-Match subroutine is invoked and the  search window is shifted by at least $|sparse(P)|$.  Hence     in all 
three types of events the expected search window shift is at least $O(min(|sparse(P)|,\Delta))$. Hence the result.
\end{proof}
\begin{lemma}
For any given pattern string $P$ of length $n$ from $\Sigma$ and a $n$ length substring of query string $T$ whose    characters 
are drawn independently and uniformly from $\Sigma$, the expected number of matches before a mismatch when invoking Random-Match subroutine is $O(1)$.
\end{lemma}
\begin{proof}
Each character in $T$ is drawn independently and uniformly from $\Sigma$. Therefore,   the  expected number of matches  before a 
mismatch = $\frac{\Delta-1}{\Delta}(1 + \frac{2}{\Delta} + \frac{3}{\Delta^2} + . . .  + \frac{n-1}{\Delta^n}) = O(1)$.
\end{proof}
\begin{lemma}
For any    pattern string $P$, the length of $sparse(P)$, the longest $2$-sparse  pattern of $P$, is at least $\delta$,   where $\delta$ is the number of distinct characters in $P$.
\end{lemma}
\begin{proof}
Let $a$ be the character in $P$ whose last   occurrence has the smallest index and let its index in $P$ be denoted by $start$.
From definition of $a$, we can observe that every character in $P$ occurs at least once to the right of index $start$ in $P$.
Let $b$ be the character in $P$ whose first occurrence in $P$ to the right of $start$ has the highest index and let its position in $P$ be denoted by $end$. We    can easily   observe that in the interval $[start, end]$ all characters in $P$ are present at least once and the characters $a$ and $b$ do not appear in between. Therefore, $|sparse(P)| \geq |sparse^{(a,b)}(P)| \ge \delta$.
\end{proof}
\begin{lemma}
For any pattern string $P$, $R$ preprocesses $P$ in $O(n\delta)$ time to determine $sparse(P)$, and  $shift^{c}(P)$, for $c \in \Sigma$, where $\delta$ is the number of distinct characters in $P$.
\end{lemma}
\begin{proof}
First, by scanning $P$ in $O(n)$ time we can compute $shift^{c}(P)$, for $c\in P$. Second, for each pair of characters $a, b 
\in P$, we initialize $sparse^{(a,b)}(P)$ to $0$. Then, as we scan $P$, when we  encounter  a character $c \in P$, (i)    we update the index of its last occurrence, and (ii) based on the index of $c$ and the position of the       last occurrence of characters $x$ in $P$ that have occurred earlier, we update $sparse^{(c,x)}$. This     requires $O(n)$ time for the scan and 
and for each character in $P$ $O(\delta)$ time for at most $\delta$ updates to the length of sparse patterns. 
Therefore, the total time during scan is $n\delta$. Finally, since   there are     $\delta^2$ ordered pairs, we
can        trivially find the maximum length for all pairs of characters in $P$ and        from them   choose the longest in $O(\delta^2)$ time. Hence, the total preprocessing time is $O(n\delta + \delta^2) = O(n \delta)$.
\end{proof}
\section{Expected Length of $sparse(P)$ for a Random String $P$}
In this section, we show that   for a pattern string $P$ whose characters are chosen uniformly at random from an alphabet of
size    $\Delta$,   $E[|sparse(P)|]$   is 
$\Omega(\Delta log (\frac{2\Delta}{2\Delta-\delta}))$. We will first present the key idea 
behind our proofs, then we introduce some definitions that are necessary for stating and establishing a       lower bound on $E[|sparse(P)|]$. \newline \newline 
{\bf Key Idea}: 
Let $\cal{P}$ be the    set of strings of length $n$  whose characters   are drawn from  the alphabet $\Sigma=\{c_1, c_2,..., 
c_{\Delta}\}$ and $\cal{P}(\delta) \subseteq \cal{P}$ be the set of strings with at least  $\delta$  distinct characters. For  
a pattern string $P$ that is chosen  uniformly at random from $\cal{P}$, we show                              $E[|sparse(P)|] 
= \frac{1}{|\cal{P}|}$   $\sum_{P \in \cal{P}}^{} |sparse(P)|$ is      $\Omega(\Delta log(\frac{2\Delta}{2\Delta-\delta+1}))$ 
by first showing $\frac{|\cal{P}(\delta)|}{|\cal{P}|}$ =   $(1-o(1))$ and then lower bounding   $\frac{1}{|\cal{P}(\delta)|}$ $\sum_{P \in \cal{P}(\delta)}^{} |sparse(P)|$ as follows:  
\begin{itemize} 
\item [(1)] We first define an onto function $F: \cal{P}(\delta) \rightarrow \cal{Q}(\delta)$ that maps a  given string $P = 
(a_1, a_2, ..., a_n ) \in {\cal P}$ to a string $Q = F(P)=(b_1, b_2, ..., b_n)$                          such that (i)  there 
are at least $\frac{\delta}{2}-1$ distinct characters between the first and second occurrence of $b_1$ in $Q$, (ii) ${\cal Q(\delta)} 
\subseteq   {\cal P(\delta)}$, and (iii) $|{\cal Q(\delta)}| \geq \frac{1}{2}|{\cal P(\delta)}|$. 
\item [(2)]    We then use properties of the mapped string $Q=F(P)$ to lower  bound             $\frac{1}{|\cal{P}(\delta)|}$  
$\sum_{P  \in \cal{P}(\delta)}^{} |sparse(P)|$ as follows:

\[
\frac{1}{|{\cal P(\delta)}|}  \sum_{P \in {\cal P(\delta)}}^{} |sparse(P)| \geq \frac{1}{2|{\cal Q(\delta)}|} 
\sum_{Q \in {\cal Q(\delta)}}^{} |sparse(Q)|
 \geq 
\frac{1}{2|{\cal Q(\delta)}|} \sum_{Q \in {\cal Q(\delta)}}^{} firstpos^{\frac{\delta}{2}+1}(Q) 
\]
\[
\geq 
\frac{1}{2}(1-o(1))E[W^{\frac{\delta}{2}+1}] = \Omega( \Delta log(\frac{2\Delta}{2\Delta-\delta+1}))
\]
\begin{tabbing} 
where, 
\= $firstpos^{\frac{\delta}{2}+1}(Q)$ \= - \= denotes the position of    the first occurrence of the $(\frac{\delta}{2}+1)^{th}$ distinct \\ \> \>                                \> character in $Q$; \\
\> $W^{\frac{\delta}{2}+1}$ \> - \> the waiting time for the selection of the $(\frac{\delta}{2}+1)^{th}$ distinct    coupon \\
\>                    \>   \> in    a   sequence of independent trials where during each trial a    coupon  \\
\>                    \>   \>  is selected uniformly at random from among $\Delta$ different coupon types.
\end{tabbing} 
\end{itemize} 
\begin{definitions} 
Let $P = (a_1, a_2, ..., a_n)$ be     an arbitrary string in $\cal{P}(\delta)$.  Let    $FIRST(P)$=($a_{\pi(1)}$,$a_{\pi(2)}$, $...,a_{\pi(\frac{\delta}{2})}$) denote the subsequence of $P$ of length $\frac{\delta}{2}$ consisting    of the first occurrence of the first $\frac{\delta}{2}$   distinct characters in $P$ and $NEXT(P) = (a_{\pi(\frac{\delta}{2}+1)}, ..., a_{\pi(\delta)})$ denote the  subsequence of 
the first occurrence of   the next $\frac{\delta}{2}$ distinct characters in $P$.  
\end{definitions}
\begin{definitions} 
Let $P=(a_1, a_2, ..., a_n)$ be an arbitrary string in $\cal{P}(\delta)$, $r= rank^{FIRST(P)}(a_1)$ denote the lexical rank of
$a_1$ among the characters in $FIRST(P)$, and $b_1=a_{\pi(\frac{\delta}{2}+r)}$ denote the first occurrence of the $(\frac{\delta}{2}+r)^{th}$ distinct       character in $P$.  We now define the function $F : {\cal P(\delta)} \rightarrow {\cal Q(\delta)}$  as  follows:
\[ Q = F(P) = \left\{ \begin{array}{ll}
					P & \mbox{if $a_1$ occurs exactly once in the substring $(a_1, a_2, ..., a_{\pi(\delta/2)})$} \\
					(b_1, a_2, ..., a_n) & \mbox{otherwise}
				  \end{array}
		  \right. \]
\end{definitions} 
\begin{theorem}
Let $\cal{P}$ be the    set of strings of length $n$  whose characters   are drawn from  the alphabet $\Sigma=\{c_1,  c_2,..., 
c_{\Delta}\}$ and $P$ be a string in ${\cal P}$ whose characters are chosen uniformly at random from  $\Sigma$.          Then, $E[|Sparse(P)|]=
min(\frac{\Delta}{2}, \Delta log(\frac{2\Delta}  {2\Delta -\delta+1}))$.
\end{theorem} 
\begin{proof} 
From Lemma $5$, we know that $|Sparse(P)| \geq \delta$. So to establish this theorem, we only need to consider   the situation 
when $\delta < \frac{\Delta}{2}$. From definition, we know $E[|Sparse(P)|] = \frac{1}{|\cal{P}|}$ $\sum_{P \in \cal{P}}^{} |sparse(P)|
$  $\geq \frac{1}{|\cal{P}|}$ $\sum_{P \in {\cal P(\delta)}}^{} |sparse(P)|$. From Lemma $8$, we know $\frac{|\cal{P}(\delta)|
}{|\cal{P}|}$ = $(1-o(1))$. Therefore, $\frac{1}{|\cal{P}|}$   $\sum_{P \in {\cal P(\delta)}}^{} |sparse(P)|$   $\geq (1-o(1)) 
\frac{1}{|{\cal P(\delta)}|}$   $\sum_{P \in {\cal P(\delta)}}^{} |sparse(P)|$. From Lemma $10$, we get       $\frac{1}{|{\cal P(\delta)}|}$   $\sum_{P \in {\cal P(\delta)}}^{} |sparse(P)|   = \Omega(\Delta log(\frac{2\Delta}{2\Delta-\delta+1}))$. Hence
we are done.
\end{proof} 
\begin{lemma}
Let $\cal{P}$ be the    set of strings of length $n$  whose characters   are drawn from  the alphabet $\Sigma=\{c_1, c_2,..., 
c_{\Delta}\}$ and $\cal{P}(\delta) \subseteq \cal{P}$ be the set of strings with at least  $\delta$  distinct characters.  If 
$\delta < \frac{\Delta}{2}$, then $\frac{|\cal{P}(\delta)|}{|\cal{P}|}$ =   $(1-o(1))$.
\end{lemma}
\begin{proof}
We  establish this lemma by showing that the number of strings of length $n$ in $\cal{P}$ with less than    $\delta$ distinct
symbols is a small fraction of $\cal{P}$. That is, we   show that $(\Delta C_{\delta}) (\frac{\delta}{\Delta})^{n} < 1/n$. On 
expanding the left hand side, taking      logarithms on both sides and then solving for $n$, we see that the above inequality 
holds for $n > log n + \delta$.  
\end{proof} 
\begin{lemma} 
Let $F: \cal{P}(\delta)      \rightarrow \cal{Q}(\delta)$ be the function defined in Definitions $4.2$. Let 
$P (a_1, a_2, ..., a_n )$ be any string in ${\cal P(\delta)}$ that is mapped to string $Q = F(P)=(b_1, b_2, ..., b_n)$. 
We will show that (i)  there 
are at least $\frac{\delta}{2}-1$ distinct characters between the first and second occurrence of $b_1$ in $Q$, (ii) ${\cal Q(\delta)} 
\subseteq   {\cal P(\delta)}$, and (iii) $|{\cal Q(\delta)}| \geq \frac{1}{2}|{\cal P(\delta)}|$. 
\end{lemma} 
\begin{proof}
From definition of $F$, we know that if the first character in $P$ does not repeat until the first occurrence of the 
$\frac{\delta}{2}$th          distinct character then $Q=P$ and Property (i) is automatically satisfied. Otherwise, we obtain $Q$ from $P$ by 
replacing the first character in $P$ by the first occurrence of the 
$\frac{\delta}{2}+r^{th}$ distinct character in $Q$, where $r$  is 
the $rank^{FIRST(P)}(a_1) \geq 1$. Hence, Property (i) is  satisfied in this situation also.
Now, we will show that   ${\cal Q(\delta)} \subseteq   {\cal P(\delta)}$. Notice from the defintion of $F$, for any $P$, $Q = F(P)$ has the same number           of distinct symbols as $P$ and the length of $Q$ is $n$, therefore $Q$ is also an element 
of ${\cal P}(\delta)$. Therefore,     ${\cal Q(\delta)}\subseteq {\cal P(\delta)}$. Also, we        can observe that from the definition  of $F$ that for each string $Q \in {\cal Q(\delta)}$ there are at most   two strings   in ${\cal P(\delta)}$ that 
map to it. In addition, for each string $Q \in {\cal Q(\delta)}$ $Q$ itself is one of the pre-images, so  $|{\cal Q(\delta)}| \geq \frac{1}{2}|{\cal P(\delta)}|$. 
\end{proof}
\begin{lemma}
 Let $F: \cal{P}(\delta) \rightarrow \cal{Q}(\delta)$ be the function defined in Definitions $1.2$. We will now establish the 
 following inequalities.  
\begin{eqnarray} 
\frac{1}{|{\cal P(\delta)}|}  \sum_{P \in {\cal P(\delta)}}^{} |sparse(P)| \geq \frac{1}{2|{\cal Q(\delta)}|} 
\sum_{Q \in {\cal Q(\delta)}}^{} |sparse(Q)| 
\geq \frac{1}{2|{\cal Q(\delta)}|} \sum_{Q \in {\cal Q(\delta)}}^{} firstpos^{\frac{\delta}{2}+1}(Q) \\
\frac{1}{|{\cal Q(\delta)}|} \sum_{Q \in {\cal Q(\delta)}}^{} firstpos^{\frac{\delta}{2}+1}(Q) = 
(1-o(1))E[W^{(\frac{\delta}{2}+1)}] \\
E[W^{(\frac{\delta}{2}+1)}] \geq \Delta log(\frac{2\Delta}{2\Delta-\delta+1})
\end{eqnarray}
\end{lemma} 
\begin{proof}
The first inequality follows from definition of $F$ and the properties of $F$ established in Lemma $9$. The second 
inequality is established as follows. First, we compute $\sum_{Q \in {\cal Q(\delta)}}^{} firstpos^{\frac{\delta}{2}+1}(
Q)$ by computing the sum of the waiting times for the selection of the $(\frac{\delta}{2}+1)^{th}$ distinct coupon over all 
possible sequences of trials, where during a trial a coupon is selected uniformly at  random from   among $\Delta$ 
different coupon types. Then, from this total we   exclude   the waiting time of those sequences that require more 
than $n$ trials to select $\delta$ distinct coupons. More formally, 
\begin{eqnarray} 
\frac{1}{|{\cal Q(\delta)}|} \sum_{Q \in {\cal Q(\delta)}}^{} firstpos^{\frac{\delta}{2}+1}(Q)   
\approx{E[W^{(\frac{\delta}{2}+1)}] - \sum_{i \in [n+1.. \infty]}^{} i* Pr(W^{(\delta)} = i)} 
\end{eqnarray} 
The     expected
waiting   time $E[W^{(\frac{\delta}{2}+1)}]$ for the selection of $(\frac{\delta}{2}+1)^{th}$ distinct coupon can be determined by 
letting $X_{i}$   denote    the number of trials following the selection of the $i^{th}$ distinct coupon until the 
$(i+1)^{th}$ distinct coupon type is selected. Then, we   can see that $W^{(\frac{\delta}{2}+1)} = \sum_{i \in [1..\frac{\delta}{2}+1]}^{} 
X_{i}$ is the waiting time until the $(\frac{\delta}{2}+1)$th coupon type is   selected. It is easy to observe that the $X_i$'s 
are geometrically distributed random variables with parameter $p_i = \frac{\Delta-i}{\Delta}$. Therefore, 
\begin{eqnarray}
E[W^{(\frac{\delta}{2}+1)}] 
= E[\sum_{i \in [1..\frac{\delta}{2}+1]}^{} X_{i}] = \Delta log (2\Delta/(2\Delta-\delta+1))
\end{eqnarray}
For $i \in [n+1..\infty]$, we bound $i* Pr(W^{(\delta)} = i)$ as follows
\begin{eqnarray} 
i* Pr (W^{\delta}) = (\Delta C_{\delta})[\delta^{i} - \delta^{i-1}] \frac{i}{\Delta^{i}} < 
\frac{\Delta^{\delta}}{\delta!} \delta^{i} \frac{i}{\Delta^{i}} = o(1) \ \ \  \forall i > \delta + log(i)
\end{eqnarray}
Now, from Equations (4), (5) and (6) we see that the Inequality (2) follows.
\end{proof} 
\section{Conclusions and Future Work}
Given a pattern string $P$ of length $n$ consisting of $\delta$ distinct characters     and a query string $T$ of length $m$, where the characters of $P$ and $T$  are drawn from an alphabet $\Sigma$ of size $\Delta$,    the {\em exact string matching} problem consists of finding all occurrences of $P$ in $T$. For this problem, we present a    randomized heuristic    that in $O(n\delta)$  time preprocesses $P$ to    identify $sparse(P)$, a  rarely occurring substring of $P$, and  then   use  it to find all occurrences of $P$ in $T$ efficiently. This    heuristic  has  an expected      search time of
 $O( \frac{m}{min(|sparse(P)|, \Delta)})$, where $|sparse(P)|$ is at least $\delta$.  We also show that for a pattern string $P$   whose characters are chosen  uniformly at random from an alphabet of size $\Delta$, $E[|sparse(P)|]$ is 
 $\Omega(\Delta log (\frac{2\Delta}{2\Delta-\delta}))$. We believe that for a large class of non-trivial pattern 
 strings, our heuristic with better analysis can yield a randomized algorithm that has sublinear run time in the worst 
 case scenario.
\section{Acknowledgements}
I would  like to thank an anonymous referee for many suggestions and for  pointing out discrepencies in the earlier draft of this    paper. In addition, I would also like to thank Professor V. Sunitha for        her help in resolving some of my type setting problems in Latex.
\section*{References}
\begin{hangref} 
\item [(1)]  MORRIS (Jr) J.H., PRATT V.R., 1970, A linear pattern-matching algorithm, Technical Report 40, University of California, Berkeley.
\item [(2)]   Harrison, Malcolm C. "Implementation of the substring test by hashing." Communications of the ACM 14.12, pp. 777-779 (1971).
\item [(3)]  	Boyer, Robert S., and J. Strother Moore. "A fast string searching algorithm." Communications of the ACM 20.10,
pp 762-772  (1977).
\item [(4)]  	Knuth, Donald E., James H. Morris, Jr, and Vaughan R. Pratt. "Fast pattern matching in strings." SIAM journal on computing 6.2 pp. 323-350 (1977).
\item [(5)]  	Horspool, R. Nigel. "Practical fast searching in strings." Software: Practice and Experience 10.6, pp. 501-506 (1980).
\item [(6)]    Galil, Zvi, and Joel Seiferas. "Time-space-optimal string matching." Journal of Computer and System Sciences 26.3, pp. 280-294, (1983). 
\item [(7)]    Apostolico, Alberto, and Raffaele Giancarlo. "The Boyer-Moore-Galil string searching strategies revisited." SIAM Journal on Computing 15.1, pp. 98-105 (1986).	
\item[(8)]    Feng, Zhu Rui, and Tadao Takaoka. "On improving the average case of the Boyer-Moore string matching algorithm." Journal of Information Processing 10.3, pp. 173-177  (1988). 	
\item[(9)]  	Karp, Richard M., and Michael O. Rabin. "Efficient randomized pattern-matching algorithms." IBM Journal of Research and Development 31.2, pp. 249-260,  (1987).
\item[(10)]  	Sunday, Daniel M. "A very fast substring search algorithm." Communications of the ACM 33.8, pp. 132-142  (1990).
\item[(11)]  	Apostolico, Alberto, and Maxime Crochemore. "Optimal canonization of all substrings of a string." Information and Computation 95.1, pp. 76-95  (1991).
\item[(12)]  	Colussi, Livio. "Correctness and efficiency of pattern matching algorithms." Information and Computation 95.2, pp. 225-251 (1991).
\item[(13)]  	Crochemore, Maxime, and Dominique Perrin. "Two-way string-matching." Journal of the ACM (JACM) 38.3, 
pp. 650-674  (1991).
\item[(14)]  	Hume, Andrew, and Daniel Sunday. "Fast string searching." Software: Practice and Experience 21.11, pp.  1221-1248  (1991).
\item[(15)]  	Smith, P. D. "Experiments with a very fast substring search algorithm." Software: Practice and Experience 21.10,  pp. 1065-1074 (1991).
\item[(16)]  	Baeza-Yates, Ricardo, and Gaston H. Gonnet. "A new approach to text searching." Communications of the ACM 35.10, 
pp. 74-82  (1992).
\item[(17)] Ricardo A. Baeza-Yates, Chapter 10: String Seraching Algorithms, Information Retrieval: Data Structures \& Algorithms
edited by William B. Frakes and Ricardo Baeza-Yates, pp. 219-240, Prentice Hall, (1992).
\item[(18)]  	Galil, Zvi, and Raffaele Giancarlo. "On the exact complexity of string matching: upper bounds." SIAM Journal on Computing 21.3, pp.  407-437  (1992).
\item[(19)]  	Lecroq, Thierry. "A variation on the Boyer-Moore algorithm." Theoretical Computer Science 92.1, pp. 119-144 (1992).
\item[(20)]  	Raita, Timo. "Tuning the boyerâ€mooreâ€horspool string searching algorithm." Software: Practice and Experience 22.10,  pp. 879-884 (1992).
\item[(21)]   Wu, Sun, and Udi Manber. "Fast text searching: allowing errors." Communications of the ACM 35.10, pp. 83-91  (1992).
\item[(22)]  	Colussi, Livio. "Fastest pattern matching in strings." Journal of Algorithms 16.2, pp. 163-189  (1994).
\item[(23)]   Crochemore, M., Czumaj, A., Gasieniec, L., Jarominek, S., Lecroq, T., Plandowski, W., \& Rytter, W.  
Speeding up two string-matching algorithms. Algorithmica, 12(4-5), pp. 247-267 (1994).
\item[(24)]   Crochemore, Maxime, and Thierry Lecroq. "Tight bounds on the complexity of the Apostolico-Giancarlo algorithm." Information Processing Letters 63.4, pp. 195-203  (1997).
\item[(25)]  Gusfield, D. "Algorithms on strings, trees, and sequences." Computer Science and Computional Biology (Cambrigde, 1999) (1997).
\item[(26)]  	Charras, Christian, Thierry Lecrog, and Joseph Daniel Pehoushek. "A very fast string matching algorithm for small alphabets and long patterns." Combinatorial Pattern Matching. Springer Berlin Heidelberg, 1998.
\item[(27)]  	Berry, Thomas, and S. Ravindran. "A Fast String Matching Algorithm and Experimental Results." Stringology. 1999.
\item[(28)]  Charras, Christian, and Thierry Lecroq. Handbook of exact string matching algorithms. London, UK:: King's College Publications, 2004.
\item[(29)]  Crochemore, Maxime, Christophe Hancart, and Thierry Lecroq. Algorithms on strings. Cambridge University Press, 2007.
\item[(30)]  Faro, Simone, and Thierry Lecroq. "The exact online string matching problem: A review of the most recent results." ACM Computing Surveys (CSUR) 45.2 (2013).
\item[(31)]  Faro, S., and T. Lecroq. "Smart: a string matching algorithm research tool". University of Catania and University of Rouen (2011).
\item[(32)] S. Divakaran, Fast Algorithms for Exact String Matching. CoRR abs/1509.09228 (2015)
\end{hangref} 
\end{document}